\documentclass[12pt, onehalfspacing]{article} % use larger type; default would be 10pt
\usepackage[T1]{fontenc}

\usepackage{libertine}

\usepackage{geometry} % to change the page dimensions
\usepackage{xcolor, graphicx} % support the \includegraphics command and options
\usepackage{tikz}
\usepackage{multicol}

\usepackage{booktabs} % for much better looking tables
\usepackage{array} % for better arrays (eg matrices) in maths
\usepackage{paralist} % very flexible & customisable lists (eg. enumerate/itemize, etc.)
\usepackage{verbatim} % adds environment for commenting out blocks of text & for better verbatim
\usepackage{subfig} % make it possible to include more than one captioned figure/table in a single float
\usepackage{amsfonts, amsmath, amssymb, amsthm, euscript}
\usepackage{mathrsfs}
\usepackage{mhsetup, mathtools}
\usepackage{thmtools}
\usepackage{halloweenmath}

\newtheorem{proposition}{Proposition}  %[section] % To number by section 0.1, 1.2, etc
\newtheorem{lemma}{Lemma}

\newtheorem{theorem}{Theorem}
\newtheorem{corollary}{Corollary}

\newtheorem{remark}{Remark}

\theoremstyle{definition}
\newtheorem{definition}{Definition}
\newtheorem{axiom}{Axiom}
\newtheorem{example}{Example}

\renewcommand\thmcontinues[1]{Continued}

\newcommand{\R}{\mathbb{R}}

\newcommand{\F}{\mathcal{F}}
\newcommand{\B}{\mathcal{B}}

\newcommand{\ra}{\rightarrow}

\newcommand{\LRA}{\Longleftrightarrow}

\DeclareMathSymbol{\Alpha}{\mathalpha}{operators}{"41}
\DeclareMathSymbol{\Beta}{\mathalpha}{operators}{"42}
\DeclareMathSymbol{\Epsilon}{\mathalpha}{operators}{"45}
\DeclareMathSymbol{\Zeta}{\mathalpha}{operators}{"5A}
\DeclareMathSymbol{\Eta}{\mathalpha}{operators}{"48}
\DeclareMathSymbol{\Iota}{\mathalpha}{operators}{"49}
\DeclareMathSymbol{\Kappa}{\mathalpha}{operators}{"4B}
\DeclareMathSymbol{\Mu}{\mathalpha}{operators}{"4D}
\DeclareMathSymbol{\Nu}{\mathalpha}{operators}{"4E}
\DeclareMathSymbol{\Omicron}{\mathalpha}{operators}{"4F}
\DeclareMathSymbol{\Rho}{\mathalpha}{operators}{"50}
\DeclareMathSymbol{\Tau}{\mathalpha}{operators}{"54}
\DeclareMathSymbol{\Chi}{\mathalpha}{operators}{"58}
\DeclareMathSymbol{\omicron}{\mathord}{letters}{"6F}

\definecolor{purple}{RGB}{85, 6,139}
\definecolor{teal}{RGB}{2,108,128}
\definecolor{lavender}{RGB}{129, 102, 122}
\definecolor{carolina blue}{RGB}{68, 157, 209}
\definecolor{phthalo blue}{RGB}{2, 8, 135}
\definecolor{purple2}{RGB}{149, 96, 219}
\definecolor{green1}{RGB}{96, 219, 117}

\usepackage{natbib}
\usepackage{hyperref}
\usepackage{nameref}
\hypersetup{
  colorlinks   = true, %Color links instead of ugly boxes
  urlcolor     = purple, %Colour for external hyperlinks
  linkcolor    = teal, %Colour of internal links
  citecolor   = gray%gray%red!40!black} %Colour of citations
}
\usepackage{fancyhdr} % This should be set AFTER setting up the page geometry
\usepackage{sectsty}
\usepackage{setspace}
\allsectionsfont{\upshape \raggedright \fontsize{13}{15} \selectfont} %  \sffamily (see the fntguide.pdf for font help)
\usepackage[nottoc,notlof,notlot]{tocbibind} % Put the bibliography in the ToC
\usepackage[titles,subfigure]{tocloft} % Alter the style of the Table of Contents

\title{\textsc{Conservative Updating}}  %%%   TITLE %%%
\author{\href{https://www.matthewkovach.com/}{Matthew Kovach}\footnote{Department of Economics, Virginia Tech.  E-mail: mkovach@vt.edu. I would like to thank Aur\'elien Baillon, Jonathan Chapman, Federico Echenique, Bart Lipman, Pietro Ortoleva, Kota Saito, Marciano Siniscalchi, Gerelt Tserenjigmid, and Leeat Yariv for helpful feedback. This paper was previously titled \emph{Sticky Beliefs: A Characterization of Conservative Updating} and is based on chapter 3 of my dissertation at Caltech.}}

\begin{document}
\maketitle

\vspace{7 mm}

\noindent{\textbf{Abstract:} This paper provides a behavioral analysis of conservatism in beliefs. I introduce a new axiom, \nameref{Dom-C}, that relaxes \nameref{DC} when information and prior beliefs ``conflict.''  When the agent is a subjective expected utility maximizer, \nameref{Dom-C} implies that conditional beliefs are a convex combination of the prior and the Bayesian posterior.  Conservatism may result in belief dynamics consistent with confirmation bias, representativeness, and the good news-bad news effect, suggesting a deeper behavioral connection between these biases.  An index of conservatism and a notion of comparative conservatism are characterized. Finally, I extend conservatism to the case of an agent with incomplete preferences that admit a multiple priors representation.  

\vspace{7 mm}

\vspace{5 mm}


\vspace{30 mm}

\pagebreak
%\end{comment}
%\begin{document}

\section{Introduction}

Many papers in both economics and psychology have documented biases in belief updating (see, for instance, \cite{camerer1995}, \cite{Kahneman1972}, \cite{ElGamal1995}).  A recurrent finding is that people tend to exhibit conservatism (\cite{Phillips1966}, \cite{edwards1982}, \cite{Aydogan2017}, and \cite{mobius2014}). A conservative updater only partially incorporates new information into her beliefs; hence she puts too much weight on her prior.  Despite its prevalence, conservatism has yet to be behaviorally founded. 

I introduce a novel behavioral postulate, \nameref{Dom-C}, to capture conservatism within the framework of conditional preferences over acts (see \cite{Savage1954} and \cite{anscombe1963}). This axiom weakens \nameref{DC} to accommodate ``preference stickiness.'' Put loosely, \nameref{Dom-C} allows for violations of \nameref{DC} only when the initial preference and the information are in conflict. Further, I show that conservative preferences are consistent with several well-known biases, including confirmation bias, the representativeness heuristic, and the good news-bad news effect.

For a deeper intuition behind \nameref{Dom-C}, consider the following hypothetical of an agent's reaction to information regarding climate change. Suppose this agent initially favors the use of coal for power, yet she concedes that using alternative energy is better if climate change is occurring. If she obtains information regarding the scientific consensus that climate change is occurring, would she now support using alternative energy over coal? If she is dynamically consistent (i.e., Bayesian), then the answer is yes because she has already conceded that alternative energy is better in that contingency.  If she is conservative, she may still prefer coal. When there is a conflict between information (climate change news) and an agent's initial preference (use coal), conservatism may result in a violation of \nameref{DC}. \nameref{Dom-C} allows for such violations. 

 \nameref{Dom-C} restricts the agent so that she may violate \nameref{DC} only in situations in which there is a conflict between the information and her initial preference. To illustrate, consider how our agent would feel if she had initially preferred alternative energy to coal. Now that the evidence for climate change is consistent with her initial preference, she must continue to prefer alternative energy. That is, no matter how conservative she is, it would be absurd for her to (i) initially support alternative energy, (ii) acknowledge that alternative energy is better if climate change is occurring, and then (iii) state that she supports coal upon receipt of information regarding the consensus that climate change is occurring. \nameref{Dom-C} rules out reversals of this form. 
  
 Suppose the agent is a subjective expected utility maximizer (SEU) and her prior beliefs are given by a probability distribution $\mu$. In this case, \nameref{Dom-C} characterizes a subjective expected utility agent whose posterior beliefs after $A$, denoted $\mu_A$, are a convex combination of the prior and the Bayesian posterior: 
\begin{equation}\label{belief}
\mu_A = \delta(A) \mu +  (1-\delta(A)) \B(\mu,A) 
\end{equation}
where $\delta(A) \in [0,1]$ and $\B(\mu,A)$ is the Bayesian update of $\mu$ given $A$.  When $\delta(A) >0$, the agent is reluctant to move away from her initial beliefs, and therefore she is conservative. When $\delta(A)=0$, she is Bayesian; when $\delta(A)=1$, she is unresponsive to the information. Consequently,  $\delta(A)$ can be interpreted as her degree of conservatism (at $A$).\footnote{Relatedly, $1-\delta(A)$ might be thought of as a measure of her confidence in the new information. This interpretation will be useful for the discussion of source dependence in \autoref{source}.} In general, a collection of conditional preferences admits a \textbf{conservative subjective expected utility} representation when posterior beliefs are given by \autoref{belief} for each event.

An important feature of the representation in \autoref{belief} is the dependence of the degree of conservatism, $\delta(A)$, on the realized event: bias may be source-dependent. Because of source dependence, conservatism is able to capture belief dynamics consistent with confirmation bias (\citet{nickerson1998}), the representativeness heuristic (\citet{Kahneman1972}), and the good news-bad news effect (\cite{Eil2011}). Although these biases have been thought of as distinct aspects of behavior, this finding suggests that there may be a core behavioral aspect of conservatism that unifies these biases.%, which important because  

I then provide an analysis of various aspects related to an agent's degree of conservatism. Intuitively, $\delta(A)$ captures how conservative an agent is (at $A$). When $\delta(A)$ is event independent, $\delta(A)=\delta$ for some $\delta \in [0,1]$, then the agent is described by a single behavioral parameter and $\delta$ is an index of conservatism.  I establish a behavioral characterization of this case through \nameref{WC}, a novel weakening of \nameref{C}.  I then provide a simple definition of comparative conservatism that rests upon the comparison of a binary act with a constant outcome. This provides an efficient method for ordering people by each one's degree of conservatism. 

I close by exploring the implications of conservatism when agents do not have precise beliefs. To do so, I consider a collection of possibly incomplete conditional preferences and impose an adapted version of \nameref{Dom-C} on these preferences, which I call \nameref{WUC}. In this setting, the agent has a set of possible beliefs and prefers an act, $f$, to another, $g$,  when $f$ provides (weakly) greater expected utility than $g$ for every possible belief. \nameref{WUC} characterizes a conceptually similar representation to \autoref{belief}:  the agent's set of posterior beliefs is derived by mixing the initial set of priors and the set of Bayesian posteriors. A crucial distinction from the subjective expected utility case is that now information may induce a genuine expansion of the set of beliefs.  Thus subjectively ambiguous information may result in a structural change in behavior; an initially subjective expected utility agent may become ambiguity averse after information.

\subsection{Outline and Related Literature}

The setup and model are presented in \autoref{model}. Behavioral foundations are presented in \autoref{axioms}. The connection between conservatism and other belief biases is discussed in \autoref{source}. The analysis of degrees of conservatism is discussed in \autoref{discussion}. I present the extension to multiple priors and incomplete preferences in \autoref{multibeliefs}. The remainder of this section discusses related literature. 

While numerous experimental papers suggest that people do not update their beliefs in a Bayesian manner, relatively few provide an axiomatic analysis of non-Bayesian updating.  \citet{epstein2006} provided one of the first axiomatic analysis of non-Bayesian updating. He utilized a setup of preferences over menus and modeled non-Bayesian updating as a temptation (\'a la \cite{Gul2001}) to modify prior beliefs in response to an interim signal. This was extended to an infinite-horizon model by \citet{Epstein2008}.  Because of the menu-preference setting, beliefs are typically dependent on both the information and the option set. More recently, \cite{Ortoleva2012} utilized the conditional preference approach to introduce the Hypothesis Testing Representation.  In his model, a decision maker may be surprised by low probability events and, in response, adopt a new prior before applying Bayes' rule.  This representation captures ``over-reaction" to information and is conceptually distinct from conservatism.\footnote{Additionally, such a decision maker always satisfies \nameref{C}, which is typically violated by conservative updating.} 

The notion of conservative updating I characterize results in violations of \nameref{DC} and \nameref{C}. While violations of these properties are well documented, the causes of these violations are still not fully understood. Additionally, violations of \nameref{C} are relatively understudied compared to violations of \nameref{DC}. Violations of \nameref{DC} and \nameref{C} were documented by \cite{Dominiak2012} in an Ellsberg-style experiment. While these violations were more frequent among ambiguity averse subjects, they also occur among ambiguity neutral subjects.\footnote{Further, all subjects report a "loss of confidence" in their choices after information, with the greatest loss of confidence occurring among subjects who violated both \nameref{DC} and \nameref{C}. This is consistent with the interpretation of $1-\delta(A)$ as the degree of confidence in the information, with lower confidence leading to more violations.} More recently, \cite{Shishkin2019} documented behavior consistent with violations of \nameref{C} in an experimental setting.

Conservative updating, as defined in \autoref{belief}, has recently been applied to cheap-talk games in \cite{Lee2019}.  They found that a conservative receiver may induce the sender to provide more accurate signals, and thus conservatism may be welfare improving. Hence, while my results show that conservatism leads to ``mistakes'' in individual choice (the agent violates both \nameref{DC} and \nameref{C}) there may be benefits from conservatism in other situations. \cite{declippel2020} extended the model of Bayesian persuasion (\cite{Kamenica2011}) to general non-Bayesian receivers (e.g., a patient who exhibits conservatism bias).

\section{Model and Setup}\label{model}

There is a (nonempty) finite set $S$ of states of the world, a collection of events given by an algebra $\Sigma$ over $S$, and a (nonempty) set of consequences $X$. Let $\F$ denote the set of functions $f:S \rightarrow X$, referred to as an act.\footnote{It is not difficult to extend to an infinite state space by assuming $\Sigma$ is a $\sigma$-algebra and $\F$ is the set of finite-valued $\Sigma$-measurable functions.}  Following a standard abuse of notation, for any $x \in X$, by $x \in \F$ I mean the constant act that returns $x$ in every state.  Lastly, for any $f, g \in \F$ and for any $A \in \Sigma$, let $fAg$ denote the act that returns $f(s)$ when $s \in A$ and returns $g(s)$ when $s \in A^c \equiv S\backslash A$.  Following the literature, I assume that $X$ is a convex subset of a vector space.\footnote{For instance, $X$ may be an interval of monetary prizes or, as in the classic \citet{anscombe1963} setting, a set of lotteries over some prize space.}  Thus, mixed acts can be defined point-wise, so that for every $f, g \in \F$ and $\lambda \in [0,1]$, by $\lambda f + (1-\lambda)g$ I mean the act that returns $\lambda f(s) + (1-\lambda)g(s)$ for each $s \in S$.    

I assume that the agent has preferences over $\F$ conditional on her information. Formally, the agent has a collection of preference relations, $\{ \succsim_A \}_{A \in \Sigma}$ over $\F$, where $\succsim_A$ are her preferences after observing $A$.  Let $\succ_A$ and $\sim_A$ represent the asymmetric and symmetric parts of $\succsim_A$.  The case when the agent has no information is represented by $\succsim_{S}:=\succsim$. For an event $A$, say that $A$ is $\succsim$-null (or simply null) if $fAg \sim g$ for any $f,g \in \F$. Otherwise, $A$ is non-null.  Let $\Sigma_{+} \subset \Sigma$ denote the collection of non-null events. 

Let $\Delta(S)$ denote the set of probability measures over $S$. Typical elements, $\mu, \pi \in \Delta(S)$ are called beliefs. For any probability $\mu$ and non-null event $A \in \Sigma_+$, define the Bayesian update of $\mu$ given $A$ by $\B(\mu,A)(B)=\frac{\mu(B\cap A)}{\mu(A)}$ for $B \in \Sigma$. Finally, for any two utility functions $u,v:X \ra \R$, say $u \approx v$ if $u$ is a positive affine transformation of $v$.

\begin{definition}Say that a collection of preferences $\{\succsim_A\}_{A \in \Sigma}$ admits a \textbf{conservative subjective expected utility} (conservative SEU) representation if there are a non-constant utility function $u:X \rightarrow \mathbb{R}$, a prior probability $\mu \in \Delta(S)$, and a function $\delta:\Sigma_+ \ra [0,1]$ such that for all $A \in \Sigma_+$, 
\[ f \succsim_A g \Longleftrightarrow \sum_{s \in S}u(f(s))\mu_A(s) \geq \sum_{s \in S}u(f(s))\mu_A(s)\] and 
\[\mu_A = \delta(A) \mu + (1-\delta(A))\mathcal{B}(\mu,A).\]
\end{definition}

\bigskip

\begin{example}\label{ConfBias}There are two payoff states, $\mathbb{P}=\{R,B\}$ and two signals $\Theta=\{r, b\}$. Let $S=\mathbb{P} \times \Theta$ and, slightly abusing notation, let $\mathbf{r}=\{(R,r),(B,r)\}$ and $\mathbf{b}=\{(R,b),(B,b)\}$ denote the events corresponding to an $r$ or $b$ signal, respectively. The prior $\mu \in \Delta(S)$ is given in \autoref{ex1} below.

\begin{table}[h]
\centering
\begin{tabular}{ | c | c | c |}
\hline $\mu$ & $r$ & $b$ \\
\hline $R$  & $4/8$   & $1/8$   \\
\hline  $B$ & $1/8$   & $2/8$ \\  \hline
\end{tabular}
\caption{Prior over $S=\mathbb{P} \times \Theta$.}\label{ex1}
\end{table}

This example has a natural interpretation as a standard signaling experiment with prior over $\mathbb{P}$ given by $\mu_{\mathbb{P}}=(\frac58,\frac38)$, and (asymmetric) signal accuracy $\sigma(r|R)=\frac{4}{5}$, $\sigma(b|B)=\frac{2}{3}$. 

If the agent admits a conservative SEU representation, posterior beliefs after $\mathbf{r}$ and $\mathbf{b}$, written in terms of the marginal belief on $R$, are
\[\mu_{\mathbf{r}}(R) = \delta(\mathbf{r})\frac{5}{8} + (1-\delta(\mathbf{r}))\frac{4}{5}\]
\[\mu_{\mathbf{b}}(R) = \delta(\mathbf{b})\frac{5}{8} + (1-\delta(\mathbf{b}))\frac{1}{3}.\]
%=\delta(\mathbf{r})\mu(R) + (1-\delta(\mathbf{r}))\mu(R\mid r)
One feature of the representation is that $\delta$ is event (i.e., signal) dependent, allowing for a rich variety of other ``biases'' to emerge. For instance, when $ \delta(\mathbf{r}) <  \delta(\mathbf{b})$, the agent's beliefs exhibit confirmation bias in addition to conservatism. That is, because her conservatism bias is greater when she receives a signal that is counter to her ``prior hypothesis'' (i.e., $\mu(R) > \frac12 > \mu(B)$) she is more reluctant to incorporate ``disconfirming information." This and other biases are further discussed in \autoref{source}. 

\end{example}

\section{Behavioral Foundations}\label{axioms}

The first axiom imposes a subjective expected utility (SEU) representation at each information set. The conditions for this are well-established in the literature. 

\begin{axiom}[Conditional SEU]\label{CsEU} For each $A \in \Sigma$, $ \succsim_A$ admits a non-degenerate subjective expected utility representation. 
\end{axiom}

Before introducing the conservatism axiom, for comparison, I first state the classic axioms of \nameref{DC} and \nameref{C}.  An excellent discussion of both axioms is provided in \citet{Ghirardato2002}, and so I will only briefly discuss them.  

\begin{axiom}[Dynamic Consistency]\label{DC} For any $A \in \Sigma$, $A$ non-null, and for all $f,g \in \F$,  $$ fAg \succsim g \implies f\succsim_A g.$$
\end{axiom} 
\begin{axiom}[Consequentialism]\label{C}  For any $A \in \Sigma$ and for all $f,g \in \F$, $$f(s)=g(s) \mbox{ for all } s \in A \implies f\sim_Ag.$$
\end{axiom}

In essence, $fAg \succsim g$ reveals that the decision maker believes she would abandon $g$ for $f$ if $A$ occurred. \nameref{DC} states that if this is the case, then $f$ must be preferred to $g$ after being told $A$ has occurred.  \nameref{C} states that whenever two acts are identical within $A$, the agent must be indifferent between them after $A$.  It is well known that \nameref{DC} and \nameref{C}, when combined with \nameref{CsEU}, are necessary and sufficient for Bayesian updating \citep{Ghirardato2002}. Hence, these must be relaxed to allow for conservatism.

To illustrate how these axioms are relaxed, recall the introductory example on the use of alternative energy versus coal. The act $f$ is ``use alternative energy,'' $g$ is ``use coal,'' and the event $A$ is ``climate change is occurring.''  The agent concedes that alternative energy is better if climate change is occurring.   Depending on the agent's initial preference between $f$ and $g$, there are two relevant cases: 
\[\text{\bf Case 1: } f \succsim g \text{ and } fAg \succsim g,\]
\[\text{\bf Case 2: } g \succ f \text{ and } fAg \succsim g.\]

In case 1, the preference for the fixed action $f$ (over $g$) is consistent with the preference for the contingent action $fAg$ (over $g$). In case 2, the preferences are in conflict. \nameref{DC} requires the agent to conclude that $f$ is better than $g$ in the conditional preference in both cases. However, intuition suggests that a conservative agent may violate \nameref{DC} in case 2.  Indeed, she does so because she does not fully believe $A$ has occurred. Similarly, a conservative agent would never violate \nameref{DC} in case 1 since she prefers $f$ to $g$ irrespective of $A$'s occurrence.  The following axiom takes this intuition as the defining behavior of conservatism.

\begin{axiom}[Dynamic Conservatism]\label{Dom-C} For any $A \in \Sigma_+$ and for all $f,g \in \F$,  
$$\begin{rcases}\mbox{(i)} &f \succsim g \\
			\mbox{(ii)} &fAg \succsim g 
\end{rcases} \implies f\succsim_A g.$$
Further, if both (i) and (ii) are strict, then $f \succ_A g$. 
\end{axiom} 

\nameref{Dom-C} requires that if the agent (i) prefers $f$ to $g$ ex-ante and (ii) forecasts that she would abandon $g$ for $f$ if $A$ occurred, $fAg \succsim g$, then she must prefer $f$ to $g$ conditional on $A$. Therefore, while \nameref{Dom-C} allows for some violations of \nameref{DC} (e.g., the climate change example), it restricts violations to cases where the agent cannot be completely sure that she is making the right decision. Hence, we can think of \nameref{Dom-C} as allowing for a form of "stickiness" or skepticism about new information. Seen through this lens, \nameref{Dom-C} can be viewed as a cautious response when the reliability of information is (subjectively) uncertain.\footnote{Another way of viewing \nameref{Dom-C} is through a two-self interpretation. Condition (i) captures a self that does not learn, while condition (ii) captures a Bayesian self. \nameref{Dom-C} states that when the selves agree, then the agent's behavior necessarily reflects this agreement. When they disagree, then \nameref{Dom-C} says nothing. The representation result, however, shows that the agent's behavior must be governed by a ``compromise'' belief.} Consequently, \nameref{Dom-C} permits an agent to state $fAg \succsim g$ and $ g \succ_A f$.

\begin{theorem}\label{stickyrep} The following are equivalent:
\begin{itemize}
\item[(i)] The collection $\{\succsim_A\}_{A\in\Sigma}$ satisfies \nameref{CsEU} and \nameref{Dom-C};
\item[(ii)] The collection $\{\succsim_A\}_{A\in\Sigma}$ admits a conservative subjective expected utility representation.
\end{itemize}
Further, if  $(u,\mu,\delta)$ and $(u',\mu',\delta')$ both represent $\{\succsim_A\}_{A\in\Sigma}$, then
(i) $u'$ is a positive affine transformation of $u$, (ii) $\mu' = \mu$, and (iii) $\delta'(A)=\delta(A)$ for all $A$ such that $A,A^c \in \Sigma_+$.
\end{theorem}

Theorem 1 shows that when conditional preferences admit a SEU representation, \nameref{Dom-C} is the precise behavioral content of conservative updating. Further, the degree of conservatism, $\delta(A)$, is uniquely pinned down at essentially every $A$.\footnote{\nameref{Dom-C} may be extended to any event $A \in \Sigma$, rather than merely the non-null events, $\Sigma_+$. When $A$ is null, $fAg \sim g$ for any $f,g$. It then follows from \nameref{Dom-C} that $\succsim=\succsim_A$. Thus, conservative agents react to null events by completely ignoring them.}

\begin{example}[continues=ConfBias] To further illustrate this result, recall our example with two payoff states $\mathbb{P}=\{R,B\}$ and two signals $\Theta=\{r, b\}$.  As illustrated in \autoref{Indifference}, \nameref{Dom-C} ensures that the lower contour set of the conditional indifference curve passing through $f$ must contain the intersection of the lower contour sets determined by (i) the initial preference (e.g., $\mu$) and (ii) the preference that corresponds to performing Bayesian updating $\B(\mu,\mathbf{r})$. This intersection is shaded in blue. In essence, conservatism pulls the conditional indifference curve to to be more aligned with the initial preference. 
 	
	\begin{figure}[h]
	\centering
          \includegraphics[height=7cm]{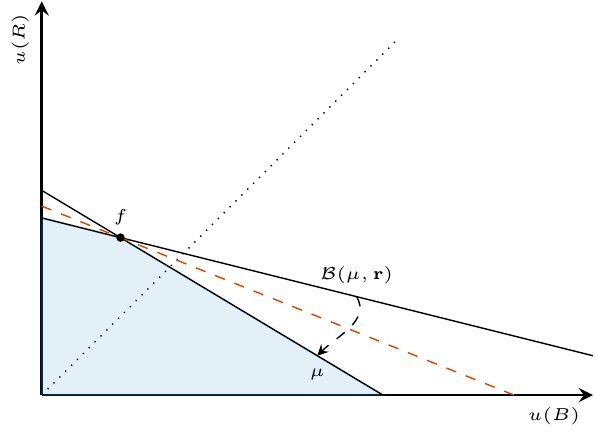}
           \caption{Indifference curves after the agent receives an $r$ signal. The (orange) dashed line corresponds to an indifference curve passing through $f$ for some $\delta(\mathbf{r}) \in (0,1)$.}\label{Indifference}
          \end{figure}
\end{example}

\begin{remark} It is easy to see from \autoref{Indifference} that over-inference or ``prior-neglect'' may be captured by rotating the (orange) indifference curve away from $\mu$ and beyond $\B(\mu,\mathbf{r})$ (e.g., intuitively, this is captured by $\delta(\mathbf{r}) < 0$). This would imply the existence of an act $g$,  such that $f \succsim g$, $f\mathbf{r}g \succsim g$ but $g \succ_{\mathbf{r}} f$, violating  \nameref{Dom-C}.  However, over-inference is difficult to capture in the conditional preference framework with a general state space. This is because $\delta(A)<0$ results in negative values for the probabilities of certain states. Thus, the behavioral foundations of such behavior remain an open question.  
\end{remark}

\section{Source Dependence and Belief Biases}\label{source}

When beliefs (or attitudes) about uncertainty vary with the source of that uncertainty, then beliefs (or attitudes) are said to be source-dependent. This section demonstrates that source-dependent conservatism captures belief dynamics consistent with confirmation bias, representativeness, and the good news-bad news effect.  Source dependence might be related to familiarity with a source. An implication of this notion is that there may be increased skepticism about new sources, which is conceptually similar to the notion of source-dependent ambiguity attitudes (\cite{Abdellaoui2011} and \cite{chew2008}).  Alternatively, source dependence might be related to message complexity, whereby subjectively simpler messages are accepted more readily.\footnote{This reasoning is consistent with theories in psychology that posit that conservatism results from the difficulty of aggregating sources of information \citep{slovic1971} and noisy recollection \citep{Hilbert2012}.}

 \subsection{Confirmation Bias}
 Confirmation bias (see \citet{nickerson1998} for an excellent review) refers to a tendency to accept information that supports already believed hypotheses and to downplay conflicting information.  Confirmation bias emerges from conservatism bias when the weight attached to the prior is greater for ``disconfirming" news than for ``confirming" news.

\begin{example}[continues=ConfBias]  There are two payoff states, $\mathbb{P}=\{R,B\}$, and two signals, $\Theta=\{r, b\}$, and the prior $\mu$ is given in \autoref{ex1} below. 

\begin{table}[h]
\centering
\begin{tabular}{ | c | c | c |}
\hline $\mu$ & $r$ & $b$ \\
\hline $R$  & $4/8$   & $1/8$   \\
\hline  $B$ & $1/8$   & $2/8$ \\  \hline
\end{tabular}
\end{table}

If the agent admits a conservative SEU representation, posterior beliefs after $\mathbf{r}$ and $\mathbf{b}$, written in terms of the marginal belief on $R$, are
\[\mu_{\mathbf{r}}(R) = \delta(\mathbf{r})\frac{5}{8} + (1-\delta(\mathbf{r}))\frac{4}{5}\]
\[\mu_{\mathbf{b}}(R) = \delta(\mathbf{b})\frac{5}{8} + (1-\delta(\mathbf{b}))\frac{1}{3}.\]

When $ \delta(\mathbf{r}) <  \delta(\mathbf{b})$, the agent's beliefs exhibit confirmation bias in addition to conservatism; when she receives a signal that is counter to her ``prior hypothesis'' (i.e., $\mu(R) > \frac12 > \mu(B)$), she is more reluctant to incorporate ``disconfirming information." I provide a behavioral characterization of confirmation bias in \autoref{GCBprop}. 
\end{example}

  \subsection{Representativeness} 
  
 The representativeness heuristic \citep{Kahneman1972} refers to a tendency to react more strongly to signals that are representative of, or similar to, an underlying state. Consider a setup like the confirmation bias example, but suppose the agent has a similarity relation $\unrhd$ where $(A,a) \unrhd (B,b)$ denotes \emph{$a$ is more representative of $A$ than $b$ is of $B$}.  If $(A,a) \unrhd (B,b)$ implies $\delta(E_a) \leq \delta(E_b),$ then a form of the representativeness heuristic is present. Beliefs depend on both the objective information provided by the signal and the degree to which the agent perceives that signal as representative of the payoff-relevant variables.

  \subsection{Good News-Bad News Effect} 

When subjects react differently to ``types'' of news, they exhibit the good news-bad news effect (\cite{Eil2011}, \cite{mobius2014}, and \cite{charness2017}).  Source-dependent conservatism may also allow for this effect. Informally, a decision maker reacts more to good news when her degree of bias is smaller.  Similar to the representativeness example, suppose $E_a \unrhd E_b$ denotes \emph{$E_a$ is better news than $E_b$}. Then asymmetric updating occurs when $E_a \unrhd E_b$ implies $\delta(E_a) \leq \delta(E_b)$. The reverse behavior, stronger reaction to bad news, is captured by reversing the inequality.

\section{Degrees of Conservatism}\label{discussion} 

\subsection{A Conservatism Index}

\autoref{stickyrep} allows for the degree of conservatism to depend on the information received. In many settings, this is useful as it allows for source-dependent reactions to news. However, it is often convenient to fully describe behavior with a single parameter, or an index of conservatism. Further, a constant degree of conservatism simplifies the task of identifying behavioral parameters from data and increases the predictive power of the model. I show that constant conservatism is characterized by a consistency condition linking conditional preferences across information sets that may be viewed as a weak form of \nameref{C}.

\begin{axiom}[Weak Consequentialism]\label{WC}  For any $A, B, C \in \Sigma_+$ with $C \cap (A\cup B) = \varnothing$ and for all $f,y,z \in \F$, \[fCy\succsim_A z \Longleftrightarrow fCy\succsim_B z.\]
\end{axiom}

To see how this is a weak form of \nameref{C}, suppose $C\cap (A\cup B)= \varnothing$ and that \nameref{C} holds. Consider $fCy$ for an arbitrary $f$.  Since for all $s \in A\cup B$, $fCy(s)=fCy(s)$, it follows from \nameref{C} that both $fCy  \sim_{A} y$ and $fCy  \sim_{B} y$. Thus, under \nameref{C}, $f$ is irrelevant, and under ordinal preference consistency, \nameref{WC} always holds.  On the other hand, \nameref{WC} does not impose indifference between $fCy$ and $y$ but requires a consistent relative preference; if $fCy$ is preferred to $z$ after $A$, then it is also preferred after $B$.

\begin{proposition}\label{constant} Suppose the collection $\{\succsim_A\}_{A\in\Sigma}$ admits a Conservative SEU representation $(u,\mu,\delta)$. Then the following are equivalent:
\begin{itemize}
\item[(i)] The collection $\{\succsim_A\}_{A\in\Sigma}$ satisfies \nameref{WC}.
\item[(ii)] There is a unique $\delta \in [0,1]$ such that $\delta(A)=\delta$ for all $A$ such that $A$ and $S\setminus A$ are non-null. 
\end{itemize}
\end{proposition}

A key strength of this result is that now $\delta$ may serve as a simple index of conservatism bias and may be elicited with only a few questions. 

\subsection{Generalized Confirmation Bias}\label{gcbsection}

 By weakening \nameref{WC} to depend on the ex-ante likelihood of events, a behavioral	characterization of conservative preferences that are consistent with confirmation bias may be obtained.  To do so, I first define a qualitative likelihood ordering over events.
 
 \begin{definition}For any $A,B \in \Sigma_+$, say that $A$ is \textbf{more likely than} $B$, denoted $A \ge_l B$ if for all $x,y \in X$, $x \succsim y$ implies $xAy \succsim xBy$. 
 \end{definition}
 
Recall that \nameref{WC} ensures a constant $\delta$ by precisely calibrating the agent's willingness to bet on $C$ across $A$ and $B$. Note that this is independent of the ex-ante relative likelihood of $A$ or $B$. Confirmation bias, on the other hand, suggests that subjectively more likely events are incorporated more accurately. In other words, because $B$ is viewed as less likely than $A$, the agent is more willing to bet on $C$ after $B$ than after $A$. This is captured by the following behavioral axiom. 
 
 \begin{axiom}[Generalized Confirmation Bias]\label{GCB}  For any $A, B, C \in \Sigma_+$ with $C \cap (A\cup B) = \varnothing$ and for all $x,y,z \in \F$ if $x\succ y$ and $A \ge_l B$, then  \[ xCy\succsim_B z \implies xCy\succsim_A z.\]
\end{axiom}
  
 \begin{proposition}\label{GCBprop} Suppose the collection $\{\succsim_A\}_{A\in\Sigma}$ admits a Conservative SEU representation $(u,\mu,\delta)$. Then the following are equivalent:
\begin{itemize}
\item[(i)] The collection $\{\succsim_A\}_{A\in\Sigma}$ satisfies \nameref{GCB}.
\item[(ii)] $\mu(A) \ge \mu(B)$ if and only if $\delta(A) \leq \delta(B)$.  
\end{itemize}
 \end{proposition}

\subsection{Comparative Conservatism}

Intuitively, a more conservative agent is less responsive to information, which can be captured in the representation by a larger weight on the prior belief. Analogously, one person is more conservative than another if he places a larger weight on his prior belief than she does on her prior belief.  This can be formally defined in terms of preferences over binary acts.

\begin{definition}\label{CompareDef}Say that $\{\succsim^1_A\}_{A\in\Sigma}$ is \textbf{more conservative} than $\{\succsim^2_A\}_{A\in\Sigma}$ if for all $A$ and all $x,y_1,y_2 \in X$ satisfying (i) $x \succ^i y_i$ and (ii) $xAy_1 \succsim^1 z \LRA xAy_2 \succsim^2 z$ for all $z \in X$  

\[xAy_1 \succsim_A^1 z \implies xAy_2 \succsim_A^2 z.\]
\end{definition}

If both agents are Bayesian, then $xAy \sim^i_A x$ for any $x,y \in X$. A conservative agent, however, may worry about the low payoff $y$ on $A^c$. The more conservative the agent, the lower his certainty equivalent for a bet on $A$. Consequently, a more conservative agent places a lower value on $xAy$ than a less conservative agent. Hence, when agent $1$ (he) is more conservative than agent $2$ (she), his certainty equivalent is lower than hers, and so whenever he prefers to bet on $A$, so must she.  Importantly, this definition dos not require the agents to have the same initial beliefs, but only the same tastes over constant acts.\footnote{By allowing $y_1$ and $y_2$ to differ, we account for the different prior beliefs about the \emph{ex-ante} likelihood of $A$. Accordingly, we may take $y_1=y_2$ if and only if $\mu_1(A)=\mu_2(A)$.}

\begin{proposition}\label{comparative} Suppose $(\{\succsim^i_A\}_{A\in\Sigma})_{i=1,2}$ admit Conservative SEU representations $(u_i,\mu_i,\delta_i)_{i=1,2}$ where $u_1\approx u_2$ and $\Sigma^1_+=\Sigma^2_+$. Then the following are equivalent:
\begin{itemize}
\item[(i)] $\{\succsim^1_A\}_{A\in\Sigma}$ is more conservative than $\{\succsim^2_A\}_{A\in\Sigma}$.
\item[(ii)] $\delta_1(A) \geq \delta_2(A)$ for every $A$ such that $A$ and $S\setminus A$ are non-null. 
\end{itemize}
\end{proposition}

\autoref{comparative} may be useful when attempting to classify subjects based on their degree of conservatism because it shows that subjects may be compared with relatively few questions. Given an event of interest, the experimenter need only elicit (conditional) certainty equivalents for particular binary acts (on $A$). Under the assumption of constant conservatism (i.e., \nameref{WC} holds), a single elicitation of a (conditional) certainty equivalent suffices to order all subjects.

\section{Multiple Beliefs}\label{multibeliefs}

Agents may struggle to come up with a single, probabilistic belief; they perceive a situation to be ambiguous.\footnote{See \cite{Gilboa2011} for an excellent summary of the literature.} Under ambiguity, we often suppose the agent has multiple beliefs.  To study conservatism with multiple beliefs, I suppose the agent has a collection of (incomplete) preferences, each of which admits a multiple-prior ``unanimity'' representation \`a la \cite{Bewley2002}. 

Similar settings have been used to study objective versus subjective rationality \citep{gilboa2010}, distinguish indecisiveness in tastes versus beliefs \citep{Ok2012}, and differentiate ambiguity perception and attitude \citep{Ghirardato2004}.  Unlike in \cite{gilboa2010}, I am focused purely on the evolution of beliefs and so do not consider the secondary ``subjective'' preference relation. Because of this my results hold irrespective of the agent's ambiguity attitude. 

\begin{definition} 
A preference relation $\succsim$ admits a multi-prior expected utility representation if  there are a utility $u:X \ra \mathbb{R}$ and a nonempty, closed, convex set of beliefs $\mathcal{M} \subseteq \Delta(S)$ such that for all acts $f,g \in \F$, 
\begin{equation}\label{mpeu}f \succsim g  \Longleftrightarrow \sum_{s \in S}u(f(s))\mu(s) \geq \sum_{s \in S}u(g(s))\mu(s) \text{ for every } \mu \in \mathcal{M}.\end{equation}
\end{definition}

 Suppose the agent has a collection of possibly incomplete preferences $\{\succsim_A^*\}_{A \in \Sigma}$ defined over $\F$.  For an event $A$, say that $A$ is unambiguously non-null if for all $x,y$ such that $x \succ^* y$, $xAy \succ^* y$.  Let $\Sigma_+^*$ denote the set of unambiguously non-null events. For any set of probabilities $\mathcal{M}$ and any $A \in \Sigma^*_+$, the set obtained when each belief in $\mathcal{M}$ is updated by Bayes' rule is denoted $\B(\mathcal{M},A)=\{\pi \in \Delta(S) \mid \pi = \B(\mu,A) \text{ for some }\mu \in \mathcal{M} \}$. Lastly, for any set $Y \subseteq \Delta(S)$, let $conv(Y)$ denote the convex hull of $Y$. 
 
\begin{axiom}[Conditional Multi-prior Expected Utility]\label{wbp} For each $A \in \Sigma^*_+$, $ \succsim_A^*$ admits a nondegenerate multi-prior expected utility representation. That is, there exists a pair $(u_A,\mathcal{M}_A)$ that represents $\succsim_A^*$ as in \autoref{mpeu}.
\end{axiom}
As in the case of subjective expected utility, the conditions for such a representation are well-known in the literature and so are not re-stated.

The following axiom, \nameref{WUC},  is precisely \nameref{Dom-C} applied to the the unambiguous preferences. 

\begin{axiom}[Unambiguous Dynamic Conservatism]\label{WUC} For any $A \in \Sigma_+^*$, and for all $f,g \in \F$
\[\begin{rcases}
 f \succsim^* g \\
 fAg \succsim^* g 
\end{rcases} \implies f \succsim_{A}^* g.\]
Further, if both (i) and (ii) are strict, then $f \succ_A^* g$. 
\end{axiom}

The following theorem is the direct counterpart of \autoref{stickyrep} for multiple priors.

\begin{theorem}\label{ambiguityrep}
The following are equivalent:
\begin{itemize}
\item[(i)]The collection $\{\succsim_A^*,\}_{A \in \Sigma}$ satisfies \nameref{wbp} and \nameref{WUC}.
\item[(ii)] There is a non-constant utility function $u:X \rightarrow \mathbb{R}$ such that $u_A=u $ for every $A \in \Sigma$, and for each $A \in \Sigma_+^*$,
\begin{equation}\label{Multi-conserv}\mathcal{M}_A \subseteq conv\left( \mathcal{M} \cup \B(\mathcal{M},A)\right).\end{equation}
In this case, we say the agent admits a \textbf{conservative multi-prior representation}. 
\end{itemize}

\end{theorem}

As before, \nameref{WUC} ensures that risk attitudes are unchanged by the information (i.e., $u$ is independent of $A$). However, \autoref{ambiguityrep} is considerably more general than \autoref{stickyrep}, which corresponds to the special case where $\succsim^*$ and $\succsim_A^*$ are complete. An important point of departure is that  \autoref{ambiguityrep} allows for an agent to \emph{expand} her set of beliefs. For instance, this occurs when $\succsim^*$ is complete but $\succsim_A^*$ is not. At the same time, my result complements results by \cite{Ghirardato2008} and \cite{Faro2019}, who found that imposing dynamic consistency on the unambiguous preference ensures prior-by-prior updating (e.g., $\mathcal{M}_A=\B(\mathcal{M},A)$).

\begin{example}[Single prior-multiple weights]\label{sm} Suppose $\mathcal{M}=\{\mu\}$ and for each $A \in \Sigma_+^*$ there exists a closed, convex set of weights $W_A \subset [0,1]$ such that 
\[\mathcal{M}_A=\{\pi \in \Delta(S) \mid \pi = \delta_A\mu + (1-\delta_A)\B(\mu,A)  \text{ and } \delta_A \in W_A\}.\]
In this case, the agent is SEU before information. However, because she is uncertain about the reliability of the information, her set beliefs expands. Her posterior beliefs are constructed via a set of ``confidence weights'' that she places on the news. In fact, it is straightforward to show that this example corresponds to precisely the case in which the initial preference is SEU.
\end{example}

\begin{corollary}Suppose $\{\succsim_A^*\}_{A \in \Sigma}$ satisfies \nameref{wbp} and \nameref{WUC} and that the initial preference, $\succsim^*$, is complete. Then for every $A \in \Sigma_+^*$ there exists a closed, convex set of weights $W_A \subset [0,1]$ such that 
\[\mathcal{M}_A=\{\pi \in \Delta(S) \mid \pi = \delta_A\mu + (1-\delta_A)\B(\mu,A)  \text{ and } \delta_A \in W_A\}.\]
\end{corollary}

\begin{example}[Multiple priors-single weight]\label{ms} Suppose for some map $\delta:\Sigma_+^* \ra [0,1]$,
\[\mathcal{M}_A= \delta(A)\mathcal{M} + (1-\delta(A))\B(\mathcal{M},A).\]
In this case, the agent has multiple beliefs, but she has a well-defined opinion of the news source. Thus she uses a single parameter to shift her beliefs after each event $A$.\end{example}

\begin{example}[continues=ms]  Suppose there are two payoff states, $\mathbb{P}=\{R,B\}$, two signals $\Theta=\{r, b\}$, and $\mathcal{M}=conv(\{\mu,\mu'\})$, which are shown in \autoref{ex2} below. 

\begin{table}[h]

%\begin{subtable}{.5\linewidth}
\centering
\begin{tabular}{ | c | c | c |}
\hline $\mu$ & $r$ & $b$ \\
\hline $R$  & $4/10$   & $2/10$   \\
\hline  $B$ & $1/10$   & $3/10$ \\  \hline
\end{tabular}
\quad
\begin{tabular}{ | c | c | c |}
\hline $\mu'$ & $r$ & $b$ \\
\hline $R$  & $3/10$   & $3/10$   \\
\hline  $B$ & $2/10$   & $2/10$ \\  \hline
\end{tabular}
\caption{Priors $\mu$ and $\mu'$.}\label{ex2}
\end{table}

Both $\mu,\mu'$ generate the same marginals over $\mathbb{P}$: $\mu(R)=\mu'(R)=3/5$. The distinction between $\mu$ and $\mu'$ is in how they treat signals. Under $\mu$ signals are informative, while under $\mu'$ they are not.  If the agent admits a conservative multi-prior representation, posterior beliefs after $\mathbf{r}$ and $\mathbf{b}$, written in terms of the belief over payoff-states $\mathbb{P}$, are
\[\mathcal{M}_{\mathbf{r}} = conv\left(\left\{\left(\frac{4-\delta(\mathbf{r})}{5},\frac{1+\delta(\mathbf{r})}{5}\right),\left(\frac35, \frac25\right) \right\}\right)  \]
\[\mathcal{M}_{\mathbf{b}} = conv\left(\left\{\left(\frac{2+\delta(\mathbf{b})}{5},\frac{3-\delta(\mathbf{b})}{5}\right),\left(\frac35, \frac25\right) \right\}\right).\]

These belief sets are illustrated in \autoref{cons_multiple} in terms of the marginal belief on $R$. 
	\begin{figure}[h]
	\centering
          \includegraphics[width=12cm]{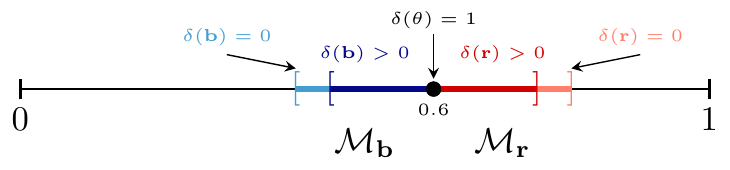}
           \caption{Posterior probabilities for payoff state $R$.}\label{cons_multiple}
          \end{figure}
Notice that $R$ is ambiguous after the signal, save for the extreme case $\delta(\theta)=1$, and the range is increasing as $\delta(\theta)$ decreases. Consequently, the conditional certainty equivalent for a bet on $R$ (e.g., a monetary act paying $\$x$ if $R$ and $\$0$ otherwise)   depends on the both the agent's degree of conservatism and her ambiguity attitude. In the case of ambiguity aversion (e.g., maxmin expected utility (MEU) representation \citep{Gilboa1989}), the agent may exhibit an ``all news is bad news'' effect. For more flexible models, such as the $\alpha$-maxmin model, this need not be the case.\footnote{By adapting axioms introduced by \cite{gilboa2010}, or more recently by \cite{Frick2020}, one may characterize an ($\alpha$-)maxmin representation consistent with $\{\succsim_A^*\}$. This additional structure on the agent's ambiguity attitude generates sharper predictions regarding her willingness to bet on various events.}  The certainty equivalents for a bet on $R$ can be seen in \autoref{ex3}. In the table, I vary $\delta(\cdot)$ and ambiguity attitude; $0$-maxmin (e.g., \cite{Gilboa1989}) is the most ambiguity averse and $1$-maxmin is maximally ambiguity seeking. Certainty equivalents for intermediate values of $\alpha$ may be directly calculated from these extreme cases. 

\begin{table}[h]

\centering
\begin{tabular}{ | c | c | c | c |}
\hline $\delta(\mathbf{r}) \backslash \alpha$ & $0$ &  $1$ \\
\hline $0$    & $0.6x$  &  $0.8x$   \\
\hline $1/2$ & $0.6x$  &  $0.7x$ \\
\hline $1$    & $0.6x$  &  $0.6x$ \\
\hline
\end{tabular}
\quad
\begin{tabular}{ | c | c | c | c |}
\hline $\delta(\mathbf{b}) \backslash \alpha$ & $0$ &  $1$ \\
\hline $0$     & $0.4x$  & $0.6x$ \\
\hline $1/2$  & $0.5x$  & $0.6x$\\
\hline $1$     & $0.6x$  & $0.6x$ \\
\hline
\end{tabular}
\caption{Certainty equivalents for the monetary bet paying $x$-utils in payoff state $R$, $0$ otherwise, for an $\alpha$-MEU agent. The left panel reports values after $\mathbf{r}$, while the right reports values after $\mathbf{b}$.}\label{ex3}
\end{table}

\end{example}

\appendix

\section{Proofs}

\subsection{Preliminary Results}

Consider the following two properties for a binary relation $\succsim$ on $\F$.
\begin{itemize} 
\item[]{\bf C-Completeness:} For any $x,y \in \F$, either $x \succsim y$ of $y \succsim x$. 
\item[]{\bf Monotonicity:} If $f(s)\succsim g(s)$ for all $s \in S$, then $f \succsim g$.
\end{itemize}

\begin{lemma}\label{OPClemma}
Consider a collection of preferences $\{\succsim_A\}_{A\in\Sigma}$ such that (i) $\succsim$ satisfies C-Completeness and Monotonicity and (ii) the collection satisfies \nameref{Dom-C}. Then for all $A \in \Sigma$ such that $A$ is non-null, and any $x,y \in X$, $x \succsim y \Longleftrightarrow x \succsim_A y.$
\end{lemma}

\begin{proof}
Since $\succsim$ is complete for constant acts, suppose $x\succsim y$.  By Monotonicity of $\succsim$ this is equivalent to $xAy \succsim y$ for all $A$. Then by \nameref{Dom-C}, $x \succsim_A y$.  Suppose that $x \succsim_A y$ but $y\succ x$.  Then it follows from Monotonicity and the fact that $A$ is non-null that $yAx \succ x$. From \nameref{Dom-C} it follows that $y \succ_A x$, a contradiction.  Hence $x \succsim y$.  
\end{proof}

\begin{lemma} Suppose $\succsim^*$ admits a multi-prior expected utility representation $(u,\mathcal{M})$. 
For each $A \in \Sigma$ and all $f,g \in \F,$ $$f Ag \succsim^* g \Longleftrightarrow fAh \succsim^* gAh \mbox{ for all } h\in\F.$$
\end{lemma}

\begin{proof}  Let $A$ be any event and let $f,g,h$ be any acts in $\F$. 
\begin{align*} f Ag \succsim^* g \Leftrightarrow & \sum_{s \in A} u(f(s))\mu(s) +  \sum_{s \in S\backslash A} u(g(s))\mu(s) \\ & \geq 
					  	 \sum_{s \in A} u(g(s))\mu(s) +  \sum_{s \in S\backslash A} u(g(s))\mu(s) \text{ for all } \mu \in \mathcal{M} \\
 \Leftrightarrow & \sum_{s \in A} u(f(s))\mu(s) \geq  \sum_{s \in A}u(g(s))\mu(s)  \text{ for all } \mu \in \mathcal{M}\\
 \Leftrightarrow & \sum_{s \in A} u(f(s))\mu(s) +  \sum_{s \in S\backslash A} u(h(s))\mu(s) \\ & \geq   \sum_{s \in A} u(g(s))\mu(s) +  \sum_{s \in S\backslash A} u(h(s))\mu(s) \text{ for all } \mu \in \mathcal{M} \\  
  \Leftrightarrow &  fAh \succsim^* gAh
\end{align*}
This lemma only relies on $A$ being non-null, and it in fact holds for more general state spaces.
\end{proof}

\begin{lemma}\label{FBlemma}
Suppose $\succsim^*$ admits a multi-prior expected utility representation $(u,\mathcal{M})$. 
For each $A \in \Sigma_+^*$ and all $f,g \in \F,$  say $f \unrhd_A^* g$ if $fAh \succsim^* gAh$ for some $h$. Then $\unrhd_A$ admits a multi-prior expected utility representation $(u,\B(\mathcal{M},A))$.
\end{lemma}

\begin{proof}  By Lemma 2, $\unrhd_A^*$ does not depend on the choice of $h$. It then follows that  $f \unrhd_A^* g$ if and only if for every $\mu \in \mathcal{M}$
\begin{align*} \sum_{s \in A} u(f(s))\mu(s) +  \sum_{s \in S\backslash A} u(h(s))\mu(s) & \geq  \sum_{s \in A} u(g(s))\mu(s) +  \sum_{s \in S\backslash A} u(h(s))\mu(s)  \\
\LRA \,  \frac{1}{\mu(A)}\sum_{s \in A} u(f(s))\mu(s) & \geq  \frac{1}{\mu(A)}\sum_{s \in A} u(g(s))\mu(s) \\
\LRA  \, \sum_{s \in A} u(f(s))\pi(s) & \geq \sum_{s \in A} u(g(s))\pi(s) \text{ for all } \pi \in \B(\mathcal{M},A),
\end{align*}
where $\B(\mathcal{M},A)=\{\pi \in \Delta(S)\mid \pi = \B(\mu,A) \text { for some } \mu \in \mathcal{M} \}$.
\end{proof}

Note that when $\succsim^*$ is complete, then we have the case of subjective expected utility and $\mathcal{M}$ and $\B(\mathcal{M},A)$ are singleton sets.

\subsection{Proof of \autoref{stickyrep}}

\begin{proof} Necessity is clear so only sufficiency is shown. By \nameref{CsEU}, there exists a $(u_A,\mu_A)$ for each $A \in \Sigma$ that represents $\succsim_A$. Further, by \nameref{Dom-C}, preferences satisfy ordinal preference consistency (see \autoref{OPClemma}): $x \succsim y$ if and only if $x \succsim_A y$. Hence we may assume $u=u_A$ for all $A$. Further, as $X$ is convex, it is without loss to suppose $[-1,1] \subset u(X)$, as $u$. 

For each $A \in \Sigma_+$, define the binary relation $\unrhd_A$ on $\F$ by $ f \unrhd_A g $ if and only if $fAg \succsim g$.  Then by \autoref{FBlemma}, $\unrhd_A$ has an expected utility representation $(u, \B(\mu,A))$.

Next, define the set $D_A:=\{\pi \in \Delta(S) \mid \delta \mu + (1-\delta) \B(\mu,A) \text{ for } \delta \in [0,1]\}$. By \nameref{Dom-C}, it follows that $\mu_A \in D_A$. Suppose not, then as $D_A$ and $\{\mu_A\}$ are closed, convex sets, there exists a separating hyperplane $a \in \R^{|S|}$ so that $\mu_A \cdot a > \hat{\mu}\cdot a $ for all $\hat{\mu} \in D_A$. Let $\bar{z}=\max_{\hat{\mu} \in D_A}\hat{\mu}\cdot a$ and let $\bar{a}=a-\bar{z}(1, \ldots,1)$. Then \begin{equation}\label{SHT1}\mu_A \cdot \bar{a} > 0 \ge \hat{\mu}\cdot \bar{a} \text{ for all } \hat{\mu} \in D_A.\end{equation}
We may suppose that $\bar{a} \in [-1,1]^{|S|}$, since we can always multiply both sides of (\ref{SHT1}) by $\epsilon>0$. Further, there are acts $f,g$ such that $u(g(s))-u(f(s))=\bar{a}(s)$ for every $s \in S$. Consequently, 
\begin{equation}\label{SHT2}\sum_{s \in S}\mu_A(s)u(g(s))  >  \sum_{s \in S}\mu_A(s)u(f(s))\end{equation}
and 
\begin{equation}\label{SHT3}\sum_{s \in S}\hat{\mu}(s)u(f(s))  \ge  \sum_{s \in S}\hat{\mu}(s)u(g(s)) \text{ for all } \hat{\mu} \in D_A.\end{equation}
By construction, $\mu, \B(\mu,A) \in D_A$ and so by (\ref{SHT3}) it follows that $fAg \succsim g$, $f\succsim g$. However, by (\ref{SHT2}) $g \succ_A f$, which contradicts \nameref{Dom-C}. Hence $\mu_A \in D_A.$ As $A$ was arbitrary, the preceding argument applies to any non-null $A$. It is standard to show that $u$ is unique up to a positive, affine transformation and, since $u(x) > u(y)$ for some $x,y \in X$, that $\mu$ and $\mu_A$ are also unique.  Given uniqueness of $\mu$ and $\mu_A$, it is trivial that there is a unique $\delta(A)$ that satisfies \autoref{belief} whenever $\mu(A)<1$.  When $\mu(A)=1$ (i.e., $A^c$ is null),  $\mu=\B(\mu,A)=\mu_A$ and $\succsim =\succsim_A$.  When $A,A^c$ are both non-null, define $\delta: \Sigma_+ \ra [0,1]$ as the unique solution to \[\mu_A=\delta(A)\mu + (1-\delta(A))\B(\mu,A).\] 
When $A^c$ is null, define $\delta(A)$ arbitrarily.  
\end{proof}

\begin{comment}

\begin{lemma}
Suppose $\{\succsim_A\}_{A\in\Sigma}$ satisfy \nameref{CsEU} and \nameref{Dom-C}. If $f(s) \succsim g(s)$ for all $s \in S$, then $f \succsim_A g$  for any $A \in \Sigma$. 
\end{lemma}

\begin{proof}
Suppose $f(s) \succsim g(s)$ for all $s \in S$.  Then it clearly follows that $fAg \succsim g$, and hence by \nameref{Dom-C} it follows that $f \succsim_A g.$ 
\end{proof}

\begin{lemma}
For each $A \in \Sigma$, $A$ non-null, it follows from \nameref{CsEU} that there is a utility index $u_A :X \rightarrow \mathbb{R}$ and probability $\mu_A$ such that $$f\succsim_A g \Longleftrightarrow  \sum_{s}u(f(s))\mu_A(ds) \geq \sum_{s}u(f(s))\mu_A(ds).$$
\end{lemma}

\begin{proof}  This follows from standard results. 
\end{proof}

\begin{lemma}
For each $A \in \Sigma$ and all $f,g \in \F,$ $$f Ag \succsim g \Longleftrightarrow fAh \succsim gAh \mbox{ for all } h\in\F.$$
\end{lemma}

\begin{proof}  
\begin{align*} f Ag \succsim g \Leftrightarrow & \sum_A u(f(s))\mu(s) +  \sum_{s\backslash A} u(g(s))\mu(s) \geq 
					  	 \sum_A u(g(s))\mu(s) +  \sum_{S\backslash A} u(g(s))\mu(s) \\
 \Leftrightarrow & \sum_A u(f(s))\mu(s) \geq  \sum_A u(g(s))\mu(s) \\
 \Leftrightarrow & \sum_A u(f(s))\mu(s) +  \sum_{S\backslash A} u(h(s))\mu(s) \geq \\
  &\sum_A u(g(s))\mu(s) +  \sum_{S\backslash A} u(h(s))\mu(s)  
 \Leftrightarrow  fAh \succsim gAh
\end{align*}
This lemma only relies on $A$ being non-null, and in fact holds for more general state spaces.
\end{proof}

\end{comment}

\subsection{Proof of \autoref{constant}}

\begin{proof}Theorem 1 shows that for each $A$, there is a $\delta(A) \in [0,1]$ satisfying the representation. Suppose \nameref{WC} holds. It is sufficient to show that for any pair of non-null events $A,B \in \Sigma$, such that both $\mu(A)<1$ and $\mu(B)<1$, $\delta(B) = \delta(A)$.  

{\bf Case 1: ($\mu(A\cup B) < 1$).} Fix any non-null $C$ satisfying $C \cap (A\cup B) = \varnothing$ and choose $x,y,z$ such that $xCy\sim_A z$. By  \nameref{WC}, it follows that $xCy\sim_B z$.  Since preferences admit a conservative SEU representation, it follows that
\begin{equation}\label{Indif1}\mu_A(C)u(x) + (1-\mu_A(C))u(y) = u(z),\end{equation} 
\begin{equation}\label{Indif2} \mu_B(C)u(x) + (1-\mu_B(C))u(y)= u(z).\end{equation}

Since $x,y$ are arbitrary, it is without loss to suppose that $u(x) > u(y)$. Then, combining (\ref{Indif1}) and (\ref{Indif2}), it is clear that $\mu_A(C)=\mu_B(C)$. Since $C \cap (A\cup B)$, it follows that $\B(\mu,A)(C)=0=\B(\mu,B)(C)$, and hence $\mu_A(C)=\delta(A)\mu(C)$ and  $\mu_B(C)=\delta(B)\mu(C)$. Hence equality is true if and only if $\delta(A) = \delta(B)$. 

{\bf Case 2: ($\mu(A\cup B)=1$).}  Suppose $A\cap B \neq \varnothing$. Then notice that $A,A\cap B$ ($B,A\cap B$) satisfy the conditions of Case 1. It follows then that $\delta(A)=\delta(A\cap B) =\delta(B).$  Now suppose $A\cap B = \varnothing$. Since there are at least three non-null events, without loss there exists some non-null set $A'$ such that $A' \subset A$ and $\mu(A')<\mu(A)$.\footnote{Alternatively, there exists $A'$ such that $A \subset A'$ and $\mu(A)<\mu(A')$, but this is just a relabeling. Note that such a pair of nested events must exist for at least one of $A$ or $B$.}  
Then we may again apply Case 1 to $A, A',A\setminus A'$, showing that $\delta(A)=\delta(A')=\delta(A\setminus A').$ Further, since $A\cap B = \varnothing$, it follows that $\mu(A' \cup B)<1$, and so $\delta(A')=\delta(B)$.
\end{proof}

\subsection{Proof of \autoref{GCBprop}} 

\begin{proof}Let $A \ge_l B$ and fix any non-null $C$ satisfying $C \cap (A\cup B) = \varnothing$. Choose any $x,y,z$ such that $x \succ y$ and suppose $xCy\sim_B z$. It follows that \begin{equation}\label{confbias1}\mu_B(C)u(x) + (1-\mu_B(C))u(y) \ge \mu_A(C)u(x) + (1-\mu_A(C))u(y).\end{equation} 
Since $u(x)-u(y) > 0 $, simplifying the above yields $\mu_B(C) \ge \mu_A(C).$ The result then follows simply from the fact that $\mu_B(C)=\delta(B)\mu(C)$ and $\mu_A(C)=\delta(A)\mu(C)$. The reverse direction is routine. \end{proof}

\subsection{Proof of \autoref{comparative}}

\begin{proof}Suppose $\{\succsim^i_A\}_{A\in\Sigma}$ admit representations $(u_i,\mu_i,\delta_i)_{i=1,2}$ where $u_1\approx u_2$. Then without loss $u_1=u_2=u$. Suppose agent $1$ is more conservative than agent $2$. Note that $A$ and $A^c$ are $\succsim_i$-non-null if and only if $\mu_i(A) \in (0,1)$. Pick some $x,y_1,y_2$ satisfying the conditions of \autoref{CompareDef}; then $u(x)\mu_1(A) + u(y_1)(1-\mu_i(A)) = u(x)\mu_2(A) + u(y_2)(1-\mu_2(A))$. 

{\bf Case 1: ($\mu_1(A)=\mu_2(A)$).}  It is without loss to suppose $y_1=y_2=y$ for some $y$ and that $u(x)=0$. Further, we may ignore the dependence of $\mu$ on $i$. If $\{\succsim^1_A\}_{A\in\Sigma}$ is more conservative than $\{\succsim^2_A\}_{A\in\Sigma}$, it follows that $\delta_1(A)(1-\mu(A))u(y) \leq \delta_2(A)(1-\mu(A))u(y)$, or $\delta_1(A) \geq \delta_2(A)$. The reverse direction is similar. 

{\bf Case 2: ($\mu_1(A)\neq\mu_2(A)$).} Note that \begin{equation}\label{compareq1}\mu_1(A)-\mu_2(A) = (1-\mu_2(A))-(1-\mu_1(A)) \neq0.\end{equation} By hypothesis, $u(x)\mu_1(A) + u(y_1)(1-\mu_i(A)) = u(x)\mu_2(A) + u(y_2)(1-\mu_2(A))$ from which, when combined with (\ref{compareq1}), it follows that  \begin{equation}\label{compareq2}(1-\mu_1(A))[u(y_1)-u(x)] = (1-\mu_2(A))[u(y_2)-u(x)].\end{equation} From $\{\succsim^1_A\}_{A\in\Sigma}$ is more conservative than $\{\succsim^2_A\}_{A\in\Sigma}$, we conclude that 
\[V^1_A(xAy_1)=\delta_1(A)[\mu_1(A)u(x) + (1-\mu_1(A))u(y_1)] + (1-\delta_1(A))u(x)\]
\[\leq \delta_2(A)[\mu_2(A)u(x) + (1-\mu_2(A))u(y_2)] + (1-\delta_2(A))u(x) = V^2_A(xAy_1). \]
Simplifying the above expression yields \[\delta_1(A)(1-\mu_1(A))u(y_1) - \delta_1(A)(1-\mu_1(A))u(x) \] \[\leq \delta_2(A)(1-\mu_2(A))u(y_2) - \delta_2(A)(1-\mu_2(A))u(x),\]
which directly implies \begin{equation}\label{compareq3}\delta_1(A)(1-\mu_1(A))[u(y_1) - u(x)] \leq \delta_2(A)(1-\mu_2(A))[u(y_2) -u(x)].\end{equation} The result that $\delta_1(A) \ge \delta_2(A)$ then follows by combining (\ref{compareq3}) with (\ref{compareq2}) and the facts that $u(y_i)-u(x) < 0$  and $0<\mu_i(A)<1$ for $i=1,2$. 

\end{proof}

\subsection{Proof of \autoref{ambiguityrep}}

\begin{proof}The proof of this theorem is similar to the proof of \autoref{stickyrep}.  First, for each $A \in \Sigma$, there exists a pair $(u_A,\mathcal{M}_A)$ such that $\succsim_A^*$ is represented by \autoref{mpeu}.  For any $A \in \Sigma_+^*$,  define the binary relation $\unrhd_A^*$ on $\F$ by $ f \unrhd_A^* g $ if and only if $fAg \succsim^* g$.  Then by \autoref{FBlemma},  $\unrhd_A^*$ has a multi-prior expected utility representation $(u, \B(\mathcal{M},A))$. 

As in the proof of \autoref{stickyrep}, since $\succsim_A^*$ is complete on constant acts for every $A$ it follows from \autoref{OPClemma} that for any $A \in \Sigma_+^*$, $x \succsim y$ if and only if $x \unrhd_A^* y$ if and only if $x \succsim_A y$. Hence it is without loss to suppose that $u=u_A$. 

Next, let $D_A:=conv(\mathcal{M}, \B(\mathcal{M},A))$. Since $\mathcal{M}$ is a closed subset of $\Delta(S)$, it is compact. Further since $A$ is unambiguously non-null, $\B(\mathcal{M},A)$ is closed and hence also compact. Further, they are both convex. Then $D_A$ is also compact and convex. 

Now, suppose for contradiction that $\mathcal{M}_A \subseteq D_A$ is false. Then there exists some $\tilde{\mu}_A \in \mathcal{M}_A \setminus D_A$.  Following an argument similar to a \autoref{stickyrep}, there exists a separating hyperplane $a \in \R^{|S|}$ so that $\tilde{\mu}_A \cdot a > \hat{\mu}\cdot a $ for all $\hat{\mu} \in D_A$. Let $\bar{z}=\max_{\hat{\mu} \in D_A}\hat{\mu}\cdot a$ and let $\bar{a}=a-\bar{z}(1, \ldots,1)$. Then \begin{equation}\label{SHT4}\tilde{\mu}_A \cdot \bar{a} > 0 \ge \hat{\mu}\cdot \bar{a} \text{ for all } \hat{\mu} \in D_A.\end{equation}
We may suppose that $\bar{a} \in [-1,1]^{|S|}$, since we can always multiply both sides of (\ref{SHT4}) by $\epsilon>0$. Further, there are acts $f,g$ such that $u(g(s))-u(f(s))=\bar{a}(s)$ for every $s \in S$. Consequently, 
\begin{equation}\label{SHT5}\sum_{s \in S}\tilde{\mu}_A(s)u(g(s))  >  \sum_{s \in S}\tilde{\mu}_A(s)u(f(s))\end{equation}
and 
\begin{equation}\label{SHT6}\sum_{s \in S}\hat{\mu}(s)u(f(s))  \ge  \sum_{s \in S}\hat{\mu}(s)u(g(s)) \text{ for all } \hat{\mu} \in D_A.\end{equation}
By construction, $\mathcal{M}, \B(\mathcal{M},A) \subset D_A$ and so by (\ref{SHT6}) it follows that $fAh \succsim^* gAh$ and $f\succsim^* g$. However, by (\ref{SHT5}) $f \not\succsim_A^* g$, which contradicts \nameref{WUC}. Hence $\mathcal{M}_A \subseteq D_A.$

\end{proof}

\bibliographystyle{ecta}

\bibliography{/Users/matthewkovach/dropbox/References/references.bib} 

\end{document}